\documentclass[10pt,doublecolumn,twoside,journal]{IEEEtran}
\IEEEoverridecommandlockouts

\usepackage{CJK}
\usepackage{flushend}

\usepackage{subfigure}
\usepackage{colortbl}
\usepackage{bm}
\usepackage{graphicx}  % Written by David Carlisle and Sebastian Rahtz
\usepackage{url}       % Written by Donald Arseneau

\usepackage{amsmath}
\usepackage{cite}

\usepackage{amsfonts,amssymb}

\usepackage{stfloats}

\usepackage{cases}
\usepackage{algorithm}

\usepackage{multirow}
\usepackage{algorithmic}
\usepackage{float}

\newtheorem{theorem}{{Theorem}}

\newtheorem{proof}{Proof}[section]

\newcommand{\ls}[1]
    {\dimen0=\fontdimen6\the\font
     \lineskip=#1\dimen0
     \advance\lineskip.5\fontdimen5\the\font
     \advance\lineskip-\dimen0
     \lineskiplimit=.9\lineskip
     \baselineskip=\lineskip
     \advance\baselineskip\dimen0
     \normallineskip\lineskip
     \normallineskiplimit\lineskiplimit
     \normalbaselineskip\baselineskip
     \ignorespaces
    }

\begin{document}

\title{Virtual Full-Duplex Wireless Communications with Zero-Interval Modulation and Sampling}
\author{
\IEEEauthorblockN{Jianyu Wang, \emph{Member}, \emph{IEEE}, Wenchi Cheng, \emph{Senior Member}, \emph{IEEE}, Wei Zhang, \emph{Fellow}, \emph{IEEE}, and Hailin Zhang, \emph{Member}, \emph{IEEE}}

\thanks{
Jianyu Wang, Wenchi Cheng, and Hailin Zhang are with State Key Laboratory of Integrated Services Networks, Xidian University, Xi'an, 710071, China (e-mails: wangjianyu@xidian.edu.cn; wccheng@xidian.edu.cn; hlzhang@xidian.edu.cn).

Wei Zhang is with School of Electrical Engineering and Telecommunications, the University of New South Wales, Sydney, Australia (e-mails: w.zhang@unsw.edu.au).

}
}
\maketitle
\thispagestyle{empty}
\begin{abstract}
In this paper, we propose a virtual full-duplex (VFD) technique with zero-interval modulation and sampling (ZIMS), where two half-duplex (HD) transceivers can simultaneously transmit signals and each transceiver can effectively receive the desired information.
In ZIMS-VFD, the transceiver inserts a zero-interval for each symbol in the transmit signal and provides self-interference (SI)-free intervals for itself. Meanwhile, it samples the receive signal in the provided SI-free intervals and restores the desired symbols. Based on orthogonal frequency division multiplexing (OFDM), we formulate the system model and show the transmit signal structure. Then, we give the transceiver design for single input single output (SISO) ZIMS-VFD and extend it to multiple input multiple output (MIMO) communications. Numerical results verify our theoretical analyses and show that ZIMS-VFD can effectively increase the capacity and approach the FD without SI.
\end{abstract}


\section{introduction}
\IEEEPARstart{F}{ull}-duplex (FD), where transceivers can transmit and receive signals simultaneously in the same frequency-band, has gained substantial attention due to the potential to double the capacity of wireless communication networks~\cite{FD_Magazine_2}. The most critical challenge for FD is to cancel the self-interference (SI), which is the interference leaked from the local transmitter to the local receiver~\cite{Review_3_R2,hu2022performance}. In practice, SI is hard to be completely canceled because it is much stronger than the desired signals from other transceivers~\cite{FD_MAC_FCTS,wan2023performance}. To achieve FD, a series of self-interference cancellation (SIC) techniques are desired to be comprehensively utilized~\cite{FD_Survey1}.

Based on the location where the SI is canceled, SIC techniques can be divided into propagation domain SIC, analog domain SIC, and digital domain SIC. Propagation domain SIC is the first step to cancel the SI. The authors of~\cite{Stanford} used two transmit antennas whose distances from the receive antenna differ by half a wavelength. Thus, the SI signals can add destructively. Also, in~\cite{Review_3_R3}, antenna design with resonant wavetraps has been used to improve electromagnetic isolation in a compact device. Moreover, in~\cite{Review_3_R8}, an antenna isolation architecture based on electrical balance in hybrid junctions has been proposed, which is based a single antenna and can be implemented on-chip.

After the propagation domain cancellation, the residual SI needs to be further canceled using a radio frequency (RF) or base-band analog canceller before the analog-to-digital converter (ADC). For examples, the authors of~\cite{Review_3_R4} presented a novel tapped delay line RF canceller with multiple non-uniform pre-weighted taps, which can not only cancel the SI from the direct transmit-to-receive antenna coupling but also the SI from reflection paths. Also, it is shown in~\cite{Analog_TWC_2019} that the global optimal solution to the weights of the multi-tap RF canceller is given. Moreover, in~\cite{Analog_TCS_2021}, SI is first canceled by an integrated hybrid in the RF front-end, which can track the antenna impedance variations, and then canceled in the base-band analog domain based on the down-converted sampling of the transmit signal. In addition, taking nonlinear distortion into account, the authors of~\cite{Review_3_R5} combined self-adaptive RF cancellation with a shared-antenna architecture and achieved 40 dB SIC. The SIC in the propagation domain and analog domain guarantees that the weak desired signal can be sampled by ADC and not concealed in quantization noise~\cite{Wenchi_INFOCOM_2013}.

After analog cancellation, SI needs to be further canceled in the digital domain. For examples, the authors of~\cite{Digital_TWC_2015} proposed to use the digital domain copy of RF transmit signal to mitigate both the SI and transceiver impairments. Also, it is shown in~\cite{Digital_TSSC_2017} that the receive beamforming is done in the digital domain to mitigate SI. In addition, the authors of~\cite{Review_3_R1} considered RF imperfections and showed that the dominant SI in digital domain after antenna isolation and RF cancellation can be modeled through a widely linear transformation of the original transmit data. Based on the model, they proposed a widely linear digital SIC scheme. We can observe from above works that to cancel the strong SI, a series of complex works in the propagation domain, analog domain, and digital domain are desired to be done. Thus, the system complexity of FD is relatively high as compared with half-duplex (HD).

To avoid the complex SIC, a feasible way is virtual full-duplex (VFD).
Two typical VFD techniques are VFD relaying and rapid on-off-division duplex (RODD).
In VFD relaying, two or more HD relays are used to assist the transmission from the source to the destination. The relays are separated in space or time and alternately receive and forward. The inter-relay interferences (IRI) are much weaker than SI in FD. Thus, the source can transmit information in every time-slot and the spectrum efficiency loss of HD can be compensated. The authors of~\cite{VFDR_TWC_20162} proposed a VFD relaying scheme where the selected best relay forwards the packet while other relays attempt to decode the new packet. Taking into account the IRI, the overall outage probability is obtained in closed-form. Also, instead of avoiding or canceling IRI, in~\cite{Orikumhi_TVT_2017} IRI is considered as an additional source of energy to the relay and exploited for energy harvesting. In addition, VFD cooperative non-orthogonal multiple access (NOMA) is investigated in~\cite{Kim_TWC_2019}, where a relay selection algorithm with adaptive IRI management is proposed and the corresponding diversity-multiplexing tradeoff performance is analyzed based on a discrete Markov chain.
However, VFD relaying technique focuses on relay channel with two or more relays, which cannot be extended to the point-to-point communications.

RODD is another VFD technique and can be used in point-to-point communications. The transceiver with RODD follows a random on-off mask to transmit and receives in the off-slots. Based on RODD, transceivers can transmit and receive simultaneously and achieve FD at the frame level. The authors of~\cite{RODD1} proposed RODD and investigated the throughput. In~\cite{RODD2}, mutual broadcasting with RODD signaling is investigated and a practical algorithm for encoding and decoding the short messages is provided. Also, the hardware implementation of RODD has been shown in~\cite{RODD3}. However, in RODD, SI leads to erasure-slots in the receive signal. The data corresponding to erasure-slots is lost during the transmission and needs to be restored based on channel coding. Additional bits are required to correct the data corresponding to erasure-slots, which results in the throughput reduction. To the best of our knowledge, it is currently unclear how to use HD transceivers to achieve point-to-point FD without losing the information of the signal with SI.

In this paper, we propose a point-to-point VFD communication technique with zero-interval modulation and sampling (ZIMS), which uses two HD transceivers to achieve FD communication. Specifically, a zero-interval is inserted for each symbol in the transmit signal of each transceiver to provide the corresponding SI-free interval for itself. Meanwhile, it samples the receive signal without SI in the provided SI-free interval and restores the desired symbol. Since ZIMS-VFD achieves FD with modulation and sampling, the hardware overheads associated with SIC, which usually include voluminous and power demanding components, can be reduced. In addition, as compared with RODD, where additional bits are required to restore the data of erasure-slots, ZIMS-VFD can restore the desired symbols from the samples in SI-free intervals without additional coding. Thus, ZIMS-VFD can achieve a higher capacity. Based on orthogonal frequency division multiplexing (OFDM), we formulate the system model and show the transmit signal structure. Then, we give the transceiver design for single input single output (SISO) ZIMS-VFD and extend it to multiple input multiple output (MIMO) communications. Moreover, numerical evaluations are conducted to verify our theoretical analyses.

The remainder of this paper is structured as follows. In Section II, we formulate the system model and show the transmit signal structure. In Section III, we give the transceiver design for SISO ZIMS-VFD. In Section IV, SISO ZIMS-VFD is extended to MIMO communications. Numerical evaluations are conducted in Section~V. We conclude in Section~VI.
\section{System Model and Transmit Signal Structure for ZIMS-VFD}
\begin{figure}
\centering
\includegraphics[width=0.5\textwidth]{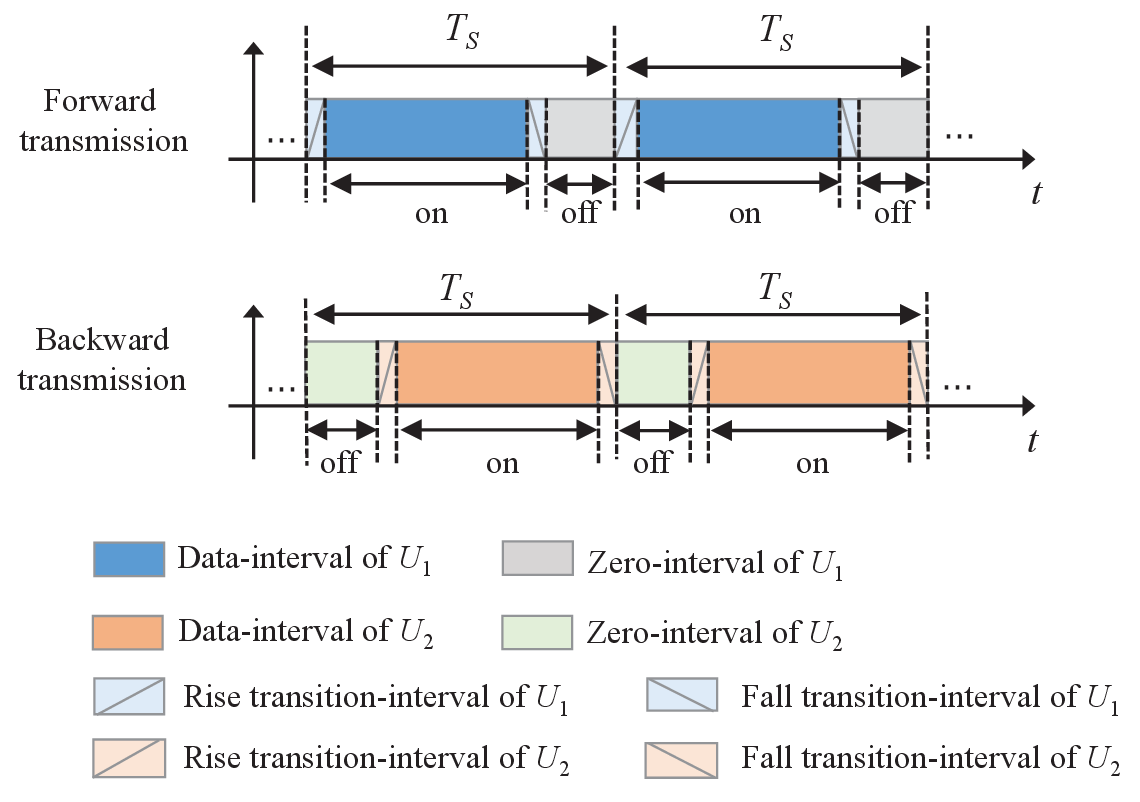}
\caption{Transmit signal structure of ZIMS-VFD.} \label{fig:Performance_Gain}
\end{figure}

We consider a point-to-point communication, where two HD transceivers, denoted by $U_1$ and $U_2$, exchange information simultaneously in the same frequency-band. The numbers of transmit antennas and receive antennas of $U_i$ are $K_{T,i}$ and $K_{R,i}$, respectively ($i\in\{1,2\}$).
We assume that the transmit signal of each transceiver is based on orthogonal frequency division multiplexing (OFDM) with $2N$ subcarriers. The frequency corresponding to the $n$-th subcarrier, denoted by $f_n$, is $f_n=f_c+n\Delta f$, where $\Delta f$ is the bandwidth of each subcarrier and $f_c$ is the center frequency ($n\in\{-N+1,...,0,...,N\}$).  It is assumed that $U_1$ and $U_2$ are well synchronized with techniques such as global positioning system (GPS)~\cite{RODD1,RODD2} and MIMO cable~\cite{RODD3}.

The transmit signal structure of ZIMS-VFD is shown in Fig.~\ref{fig:Performance_Gain}, where the forward transmission (from $U_1$ to $U_2$) and the backward transmission (from $U_2$ to $U_1$) are overlapped in the time domain.
Each transceiver uses RF switches to insert a zero-interval for each OFDM symbol period in the transmit signal. In practice, RF switching is not instantaneous and results in transitions, including rise transitions and fall transitions. Rise transitions are the switching periods from off to on. Fall transitions are the switching periods from on to off. We assume both rise transitions and fall transitions are less than a value $\delta$.
The OFDM symbol period of $U_i$, denoted by $T_{S}$, is designed to have a data-interval, denoted by $T_{D}=1/\Delta f$, a rise transition-interval, denoted by $T_{R}$, a fall transition-interval, denoted by $T_{F}$, and a zero-interval, denoted by $T_{Z}$. The expression of $T_{S}$ is given by
\begin{equation}
\begin{split}\label{T_S}
T_{S}=T_{D}+T_{R}+T_{F}+T_{Z}.
\end{split}
\end{equation}
We set $T_R=T_F=\delta$ to guarantee transitions are always within rise transition-intervals or fall transition-intervals and have no impact on zero-intervals and data-intervals.
In the zero-interval, the transmit chain of each antenna is off, which can provide an SI-free interval for local receivers.
 In the SI-free interval, the receive chain of each antenna turns on and samples the signal without SI. In the SI interval, the receive chain of each antenna turns off and shields the SI. Also, as shown in Fig.~\ref{fig:Performance_Gain}, the transmit signal structures of $U_1$ and $U_2$ are different. The data-interval of $U_1$ is before its zero-interval and the data-interval of $U_2$ is after its zero-interval. The structure difference aims to guarantee that the desired signals are non-zero in the SI-free interval.
Note that if $T_{R}+T_{F}+T_{Z}>T_{D}$, more than half of the transmit signal carries no information and ZIMS-VFD brings no symbol rate gain as compared to traditional HD. Thus, to guarantee a relatively high symbol rate, $T_D$ needs to satisfy
\begin{equation}\label{T_Z_TD}
\begin{split}
T_{D}>T_{Z}+2\delta.
\end{split}
\end{equation}

The channel model we employ is block fading. We denote by $\delta(t)$ the impulse function and ${L_{i,k,p,q}}$ the number of paths corresponding to the channel from the $p$-th transmit antenna of $U_i$ to the $q$-th receive antenna of $U_k$ ($k\in\{1,2\}$, $p\in\{1,2,...,K_{T,i}\}$, and $q\in\{1,2,...,K_{R,k}\}$). Then, based on the channel model, we can write the impulse response of the channel from the $p$-th antenna of $U_i$ to the $q$-th antenna of $U_k$ as follows~\cite{Fundamentals_of_Wireless_Communication,Digital_Communications}:
\begin{equation}\label{h_ij}
\begin{split}
h_{i,k,p,q}(t)=\!\!\!\!\sum_{l=1}^{L_{i,k,p,q}}\!\!\!\!a_{i,k,p,q}^{l}\delta(t-\tau_{i,k,p,q}^{l}),
\end{split}
\end{equation}
where $a_{i,k,p,q}^{l}$, and $\tau_{i,k,p,q}^{l}$ denote the amplitude gain and the delay of the $l$-th path, respectively.

If $ T_{Z}$ is less than the maximum delay spread of SI channels, each transceiver cannot provide SI-free intervals for the local receiver since they are covered by the local transmit signal. On the other hand, if $ T_Z$ is less than the maximum delay spread of channels from $U_1$ ($U_2$) to $U_2$ ($U_1$), inter symbol interference (ISI) occurs. Thus, to avoid ISI and provide SI-free intervals for the local receiver, $T_{Z}$ needs to satisfy
 \begin{equation}\label{T_Z_Condition}
\begin{split}
T_Z>\tau_{\max},
\end{split}
\end{equation}
where $\tau_{\max}=\max_{i,k,p,q,l}\{\tau_{i,k,p,q}^{l}\}-\min_{i,k,p,q,l}\{\tau_{i,k,p,q}^{l}\}$.
The inserted zero-intervals in the transmit signal are designed to not only help avoid SI but also inter symbol interference (ISI). Thus, to guarantee the spectrum efficiency, we don't insert additional cyclic prefixes (CPs) into the transmit signal. %Note
In addition, we set the bandwidth of each subcarrier to satisfy $\Delta f<1/\tau_{\text{max}}$, which is equivalent to \begin{equation}\label{T_D_Condition}
\begin{split}
T_D>\tau_{\max},
\end{split}
\end{equation}
such that each subchannel is flat.
\section{Transceiver Design for SISO ZIMS-VFD}
In this section, we first introduce the transceiver design for SISO ZIMS-VFD. The extension of ZIMS-VFD to MIMO communications is discussed in Section IV.
We consider a block with $M$ OFDM symbol periods. The transmit symbol corresponding to the $n$-th subcarrier of $U_i$ in the $m$-th OFDM symbol period is denoted by $X_{i}^{n,m}$. Then, the transmit signal of $U_1$ in the $m$-th OFDM symbol period, denoted by $x_{1}^m(t)$, can be given by
\begin{equation}\label{u_t1}
 x_{1}^m(t)=\left\{\begin{array}{l}
\begin{split}
&\!\!\!\!R_1^m(t),~ mT_S\!\leq \!t\!\!< \!mT_S\!+\!\delta;
\\
&\!\!\!\!\!\!\sum_{n=-N+1}^{N}\!\!\!\!\!\!\!X_{1}^{n,m}e^{j2\pi f_nt},~\! mT_S\!\!+\!\delta\!\!\leq\!\! t\!\!<\!\! mT_S\!\!+\!\delta\!+\!\!T_D;
\\
&\!\!\!\!F_1^m(t),~ mT_S\!\!+\!\delta\!+\!\!T_D\!\leq \!\!t\!\!<\!\! (m\!\!+\!\!1)T_S\!\!-\!\!T_Z;
\\
&\!\!\!\!0,~ (m\!\!+\!\!1)T_S\!\!-\!\!T_Z\!\leq\! t\!\!<\! (m\!+\!1)T_S,
\end{split}
\end{array}\right.
\end{equation}
where $R_1^m(t)$ and $F_1^m(t)$ are rise and fall transitions in the $m$-th OFDM symbol period of $U_1$. Similarly, the transmit signal of $U_2$ in the $m$-th OFDM symbol period, denoted by $x_{2}^m(t)$, can be given by
\begin{equation}\label{u_t2}
 x_{2}^m\!(t)\!\!=\!\!\!
\left\{\begin{array}{l}
\begin{split}
&\!\!\!\!0,~ mT_S\!\leq\!t\!< \!mT_S\!+\!T_Z;
\\
&\!\!\!\!R_2^m(t),~ mT_S\!+\!T_Z\!\leq\!t\!< \!mT_S\!+\!T_Z\!+\!\delta;
\\
&\!\!\!\!\!\!\!\sum_{n=-N+1}^{N}\!\!\!\!\!\!X_{2}^{n,m}\!e^{j2\pi f_nt},mT_S\!\!+\!\!T_Z\!\!+\!\delta\!\leq\!\!t\!\!< \!\!(m\!\!+\!\!1)T_S\!\!-\!\delta;
\\
&\!\!\!\!F_2^m(t),~ (m\!+\!1)T_S\!\!-\!\delta\!\leq\!\!t\!\!<\!\! (m\!\!+\!\!1)T_S,
\end{split}
\end{array}\right.
\end{equation}
where $R_2^m(t)$ and $F_2^m(t)$ are rise and fall transitions in the $m$-th OFDM symbol period of $U_2$.
Based on \eqref{h_ij}, \eqref{u_t1}, and \eqref{u_t2}, the signal at the receive antenna of $U_i$, denoted by $y_{i}(t)$, can be given by
\begin{equation}\label{y_2k_flat_ofdm}
\begin{split}
y_{i}(t)\!\!&=\!\!\underbrace{\sum_{m=1}^M\!x_{i}^m(t)\!*\!h_{i,i}(t)}_{\text{SI}}\!+\!\!\!\underbrace{\sum_{m=1}^M\!\!x_{3-i}^m(t)\!\!*\!h_{3-i,i}(t)}_{\text{Desired signal}}\!+\!z_{i}(t),
\end{split}
\end{equation}
where * denotes the convolution and $z_{i}(t)$ denotes the noise of $U_i$. Since SISO system is considered, the antenna subscripts $p$ and $q$ are omitted in this section.
Then, the receive signal given by \eqref{y_2k_flat_ofdm} is processed in SI-free intervals to restore the desired symbols. The details are shown in the next two subsections.
\subsection{Candidate Intervals Determination}
To avoid SI, the receive chain of each transceiver turns on in SI-free intervals and turns off in SI intervals, as shown in Fig.~\ref{fig:Receiver}. RF switching also results in transitions in the receive chain. Similar to the transmit chain, we assume that they are less than $\delta$. To restore the desired symbols, \eqref{y_2k_flat_ofdm} needs to be sampled in the intervals where the SI is zero, the transitions of the receive chain are zero, and the desired signal is within the data-interval. We call this kind of intervals candidate intervals and determine them in the following.
\begin{figure}
\centering
\includegraphics[width=0.44\textwidth]{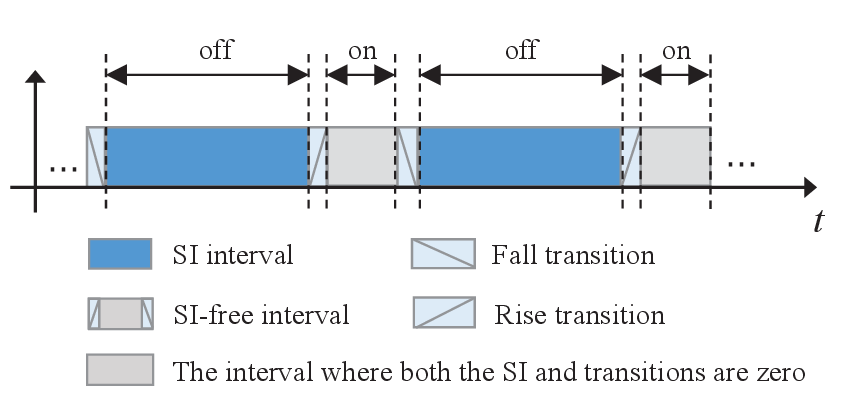}
\caption{RF switching in the receive chain of ZIMS-VFD.} \label{fig:Receiver}
\end{figure}

To determine the candidate intervals, we first investigate the SI-free intervals, where the SI is zero, for each transceiver.
Based on \eqref{h_ij}, \eqref{u_t1}, \eqref{u_t2}, and \eqref{y_2k_flat_ofdm}, the $m$-th SI-free intervals of $U_1$ and $U_2$, denoted by $\mathcal{T}_{F,1}^m$ and $\mathcal{T}_{F,2}^m$, can be expressed as
\begin{equation}\label{T_F1}
\begin{split}
\mathcal{T}_{F,1}^m\!=\!\left\{t|(m\!\!+\!\!1)T_S\!\!-\!\!T_Z\!+\!\tau_{1,1}^{\max}\!\leq\! t\!\!<\! (m\!+\!1)T_S\!+\!\tau_{1,1}^{\min}\right\}
\end{split}
\end{equation}
and
\begin{equation}\label{T_F2}
\begin{split}
\mathcal{T}_{F,2}^m=\left\{t|mT_S\!+\!\tau_{2,2}^{\max}\!\!\leq\!t\!< \!mT_S\!+\!T_Z\!+\!\tau_{2,2}^{\min}\right\},
\end{split}
\end{equation}
respectively, where $\tau_{i,k}^{\max}\!=\!\max_{l}\{\tau_{i,k}^l\}$ and $\tau_{i,k}^{\min}\!=\!\min_{l}\{\tau_{i,k}^l\}$. The RF switches corresponding to the receive chains of $U_1$ and $U_2$ start to turn on at the left ends of $\mathcal{T}_{F,1}^m$ and $\mathcal{T}_{F,2}^m$, that is, $(m\!\!+\!\!1)T_S\!\!-\!\!T_Z\!+\!\tau_{1,1}^{\max}$ and $mT_S\!+\!\tau_{2,2}^{\max}$, respectively.
Before $(m\!\!+\!\!1)T_S\!\!-\!\!T_Z\!+\!\tau_{1,1}^{\max}$ and $mT_S\!+\!\tau_{2,2}^{\max}$, the receive chains of $U_1$ and $U_2$ are in the off state, such that they can shields the SI before the left ends of $\mathcal{T}_{F,1}^m$ and $\mathcal{T}_{F,2}^m$. After a period of time $\delta$, transitions disappear and the receive chains will be in the on state. On the other hand, the RF switches corresponding to the receive chains of $U_1$ and $U_2$ start to turn off at $(m\!+\!1)T_S\!+\!\tau_{1,1}^{\min}-\delta$ and $mT_S\!+\!T_Z\!+\!\tau_{2,2}^{\min}-\delta$, respectively.
After a period of time $\delta$, the receive chains will be in the off state and transitions disappear before the right ends of $\mathcal{T}_{F,1}^m$ and $\mathcal{T}_{F,2}^m$. Thus, the SI after the right ends of $\mathcal{T}_{F,1}^m$ and $\mathcal{T}_{F,2}^m$ can be shielded. We denote by $\mathcal{T}_{FT,i}^m$ the intervals where both the SI and receive transitions of $U_i$ are zero. Then, $\mathcal{T}_{FT,1}^m$ and $\mathcal{T}_{FT,2}^m$ can be expressed as
\begin{equation}\label{T_F1}
\begin{split}
\mathcal{T}_{FT,1}^m\!\!=\!\!\left\{t|(m\!\!+\!\!1)T_S\!\!-\!\!T_Z\!+\!\tau_{1,1}^{\max}\!\!+\!\delta\!\leq\! t\!\!<\! (m\!+\!1)T_S\!+\!\tau_{1,1}^{\min}\!\!-\!\delta\right\}
\end{split}
\end{equation}
and
\begin{equation}\label{T_F2}
\begin{split}
\mathcal{T}_{FT,2}^m=\left\{t|mT_S\!+\!\tau_{2,2}^{\max}\!\!+\!\delta\!\!\leq\!t\!< \!mT_S\!+\!T_Z\!+\!\tau_{2,2}^{\min}\!\!-\!\delta\right\},
\end{split}
\end{equation}
respectively.
Next we investigate the data-interval corresponding to the desired signal of $U_i$.
Based on \eqref{h_ij}, \eqref{u_t1}, \eqref{u_t2}, and \eqref{y_2k_flat_ofdm}, the $m$-th data-intervals corresponding to the desired signals of $U_1$ and $U_2$, denoted by $\mathcal{T}_{D,1}^m$ and $\mathcal{T}_{D,2}^m$, can be given by
\begin{equation}\label{T_N1}
\begin{split}
\mathcal{T}_{D,1}^m\!\!=\!\!\left\{t|mT_S\!\!+\!\!T_Z\!\!+\!\!\delta\!\!+\!\!\tau_{2,1}^{\min}\!\!\leq\!\!t\!\!< \!\!(m\!\!+\!\!1)T_S\!\!-\!\!\delta\!\!+\!\!\tau_{2,1}^{\max}\right\}
\end{split}
\end{equation}
and
\begin{equation}\label{T_N2}
\begin{split}
\mathcal{T}_{D,2}^m\!\!=\!\!\left\{t|mT_S\!\!+\!\!\delta\!\!+\!\!\tau_{1,2}^{\min}\leq\!\! t\!\!<\!\! mT_S\!\!+\!\!\delta\!\!+\!\!T_D\!\!+\!\!\tau_{1,2}^{\max}\right\},
\end{split}
\end{equation}
respectively.

We denote by $\mathcal{T}_{C,i}^m\!=\!\mathcal{T}_{FT,i}^m\cap \mathcal{T}_{D,i}^m$ the $m$-th candidate interval of $U_i$, where the desired signal of $U_i$ is within the data-interval and the SI of $U_i$ is zero.
Then, $\mathcal{T}_{C,i}^m$ can be expressed as
\begin{equation}\label{TT_R1}
\begin{split}
\mathcal{T}_{C,i}^m=\left\{t|\chi_{L,i}^{m}\leq t\leq \chi_{U,i}^{m}\right\}
\end{split}
\end{equation}
with
\begin{equation}\label{XU1}
\begin{split}
\chi_{U,1}^{m}\!=\!\min\left\{(m\!+\!1)T_S\!+\!\tau_{1,1}^{\min}\!\!-\!\delta,(m\!\!+\!\!1)T_S\!\!-\!\delta\!+\!\!\tau_{2,1}^{\max}\right\},
\end{split}
\end{equation}
\begin{equation}\label{XL1}
\begin{split}
\chi_{L,1}^{m}\!\!=\!\!\max\left\{\!(m\!\!+\!\!1)T_S\!\!-\!\!T_Z\!+\!\tau_{1,1}^{\max}\!\!+\!\delta,mT_S\!\!+\!\!T_Z\!\!+\!\delta\!+\!\tau_{2,1}^{\min}\!\right\},
\end{split}
\end{equation}
\begin{equation}\label{XU2}
\begin{split}
\chi_{U,2}^{m}\!=\!\min\left\{mT_S\!+\!T_Z\!+\!\tau_{2,2}^{\min}\!\!-\!\delta,mT_S\!\!+\!\!\delta\!\!+\!\!T_D\!\!+\!\!\tau_{1,2}^{\max}\right\},
\end{split}
\end{equation}
and
\begin{equation}\label{XL2}
\begin{split}
\chi_{L,2}^{m}=\max\left\{mT_S\!+\!\tau_{2,2}^{\max}\!+\!\delta,mT_S\!\!+\!\!\delta\!\!+\!\!\tau_{1,2}^{\min}\right\}.
\end{split}
\end{equation}
We denote by ${T}_{C,i}^m=\chi_{U,i}^{m}-\chi_{L,i}^{m}$ the length of the $m$-th candidate interval of $U_i$
and have the following theorem.
\begin{theorem}
If $T_Z$ satisfies
\begin{equation}\label{thm}
\begin{split}
T_Z>2\delta+\tau_{\max},
\end{split}
\end{equation}
${T}_{C,i}^m>0$ can always be guaranteed.
\end{theorem}
\begin{proof}
Please refer to Appendix A.
\end{proof}

\textit{Remarks on Theorem 1:} Theorem 1 shows that if \eqref{thm} is satisfied, the transmit signal design, that is, the data-interval of $U_1$ is before its zero-interval and the data-interval of $U_2$ is after its zero-interval, can provide a candidate interval with a length larger than zero for each desired OFDM symbol. Theorem 1 proves the existence of each candidate interval and the feasibility of ZIMS-VFD to simultaneously avoid the SI and sample the desired signal in data-interval.
Combining \eqref{T_Z_TD}, \eqref{T_Z_Condition}, \eqref{T_D_Condition}, and \eqref{thm}, we can obtain the condition that $T_Z$ and $T_D$ need to satisfy:
\begin{equation}\label{Condition_Total}
\begin{split}
\tau_{\max}+2\delta<T_Z<T_D-2\delta.
\end{split}
\end{equation}

\subsection{Sampling and Desired Symbols Restoration}
If \eqref{Condition_Total} holds, candidate intervals exist. In the $m$-th candidate interval of $U_i$, the SI is zero and the signal at the receive antenna of $U_i$ is
\begin{equation}\label{y_2k_flat_ofdm_ci}
\begin{split}
y_{i}(t)\!&=\!\!\!\!\!\!\!\!\sum_{n=-N+1}^{N}\!\!\!\!\!H_{3-i,i}^{n}X_{3-i}^{n,m}e^{j2\pi f_nt}+z_{i}(t),
\end{split}
\end{equation}
where $H_{i,k}^{n}$ denotes the frequency-domain channel gain with the expression of
\begin{equation}\label{H_ijpk_New}
\begin{split}
H_{i,k}^{n}\!=\!\sum_{l=1}^{L_{i,k}}a_{i,k}^le^{-j2\pi f_n\tau_{i,k}^l}.
\end{split}
\end{equation}
Then, ${y}_i(t)$ is sampled $G$ times in $\mathcal{T}_{C,i}^m$ ($G>2N$). We denote by $\hat t_{i}^{v,m}$ the $v$-th sampling time in $\mathcal{T}_{C,i}^m$.
The expression of $\hat t_{i}^{v,m}$ can be expressed as
\begin{equation}\label{t_i_new}
\begin{split}
\hat{t}_{i}^{v,m}=\chi_{L,i}^{m}+\frac{{T}_{C,i}v}{2N}.
\end{split}
\end{equation}
Based on the samples in each $\mathcal{T}_{C,i}^m$, we can obtain the receive signal in the digital domain.
The sampling sequence in $\mathcal{T}_{C,i}^m$ can be derived based on \eqref{y_2k_flat_ofdm_ci} and \eqref{H_ijpk_New} as follows:
\begin{equation}\label{Y_ik_new}
\begin{split}
\mathbf{Y}_{i,m}=\mathbf{{V}}_{i,m}\mathbf{{H}}_{3-i,i}\mathbf{X}_{3-i,m}+\mathbf{Z}_{i,m}.
\end{split}
\end{equation}
In \eqref{Y_ik_new},
\begin{equation}\label{Y_3-i}
\mathbf{Y}_{i,m}=\left[Y_{i}^{1,m},Y_{i}^{2,m},...,Y_{i}^{G,m}\right]^T,
\end{equation}
where the $v$-th entry is the value of the $v$-th sampling point,
\begin{equation}\label{X_3-i}
\mathbf{X}_{3-i,m}=\left[X_{3-i}^{-N+1,m},X_{3-i}^{-N+2,m},...,X_{3-i}^{N,m}\right]^T
\end{equation}
is the transmit symbol vector of $U_{3-i}$, where $\left(\cdot\right)^T$ denotes the transpose operation,
\begin{equation}\label{matrix_H_i}
\mathbf{{H}}_{3-i,i}\!=\!\text{diag}\left[{H}_{3-i,i}^{-N+1},{H}_{3-i,i}^{-N+2},...,{H}_{3-i,i}^{N}\right]
\end{equation}
\begin{figure*}[ht]
\begin{equation}\label{V_i}
\begin{split}
\mathbf{{V}}_{i,m}&=\left[\begin{array}{cccc}
e^{j2\pi f_{-N+1}\hat{t}_{i}^{1,m}} &e^{j2\pi f_{-N+2}\hat{t}_{i}^{1,m}}&  \cdots & e^{j2\pi f_{N}\hat{t}_{i}^{1,m}} \\
e^{j2\pi f_{-N+1}\hat{t}_{i}^{2,m}} &e^{j2\pi f_{-N+2}\hat{t}_{i}^{2,m}}&  \cdots & e^{j2\pi f_{N}\hat{t}_{i}^{2,m}} \\
\vdots  &\vdots& \ddots & \vdots \\
e^{j2\pi f_{-N+1}\hat{t}_{i}^{G,m}} &e^{j2\pi f_{-N+2}\hat{t}_{i}^{G,m}}&  \cdots & e^{j2\pi f_{N}\hat{t}_{i}^{G,m}}
\end{array}\right]
\end{split}
\end{equation}
\hrulefill
\end{figure*}
\begin{figure*}[ht]
\begin{equation}\label{C_S}
\begin{split}
C_{S}=\frac{\Delta f}{(1+\alpha)M}\sum_{m=1}^{M}\left[\log _2 \det\left(\mathbf{I}_G+\frac{\mathbf{{V}}_{1,m}\mathbf{{H}}_{2,1}\mathbf{R}_{2,m}\mathbf{{H}}_{2,1}^H\mathbf{{V}}_{1,m}^H}{\sigma_0^2}\right)+\log _2 \det\left(\mathbf{I}_G+\frac{\mathbf{{V}}_{2,m}\mathbf{{H}}_{1,2}\mathbf{R}_{1,m}\mathbf{{H}}_{1,2}^H\mathbf{{V}}_{2,m}^H}{\sigma_0^2}\right)\right]
\end{split}
\end{equation}
\hrulefill
\end{figure*}is the channel matrix from $U_{3-i}$ to $U_i$,
\begin{equation}\label{Z_i}
\mathbf{Z}_{i,m}\!\!=\!\!\left[\bar{z}_{i}\left(\hat t_{i}^{1,m}\right),\bar{z}_{i}\left(\hat t_{i}^{2,m}\right),...,\bar{z}_{i}\left(\hat t_{i}^{G,m}\right)\right]^T
\end{equation}
is the noise vector, where each entry is assumed to follow $\mathcal C \mathcal N(0,\sigma_0^2)$ with $\sigma_0^2$ representing the variance of the noise, $\mathbf{ V}_{i,m}\in\mathbb{C}^{G\times2N}$ is related to the sampling times and subcarrier frequencies, which is given by \eqref{V_i} at the top of this page. In this paper we call $\mathbf{ V}_{i,m}$ sampling matrix. It can be observed from \eqref{Y_ik_new} that the equivalent channel matrix is $\mathbf{{V}}_{i,m}\mathbf{{H}}_{3-i,i}$. We have the following theorem for $\mathbf{{V}}_{i,m}\mathbf{{H}}_{3-i,i}$.
\begin{theorem}
$\text{Rank}(\mathbf{{V}}_{i,m}\mathbf{{H}}_{3-i,i})=2N$, where $\text{Rank}(\cdot)$ is the operation of taking the rank.
\end{theorem}
\begin{proof}
Please refer to Appendix B.
\end{proof}

\textit{Remarks on Theorem 2:} Theorem 2 implies that $\mathbf{{V}}_{i,m}\mathbf{{H}}_{3-i,i}$ is full-rank and can be transformed into $2N$ parallel subchannels. Thus, sampling in candidate intervals can achieve the same freedom degrees as conventional OFDM, where the sampling points are uniformly distributed in the entire OFDM symbol period.

In ZIMS-VFD, the sampling points are not uniformly distributed in the entire OFDM symbol period. Thus, it is difficult to directly restore $\mathbf{X}_{3-i,m}$ with discrete Fourier transform (DFT). We can consider $\mathbf{Y}_{i,m}$ as an output signal of a multiple input multiple output (MIMO) system. Based on \eqref{Y_ik_new}, the sum capacity of SISO ZIMS-VFD can be given by \eqref{C_S}. In \eqref{C_S}, $\mathbf{I}_G$ is a $G\times G$ identity matrix, $\mathbf{R}_{3-i,m}=\mathbb{E}\{\mathbf{X}_{3-i,m}\mathbf{X}_{3-i,m}^H\}$, where $\mathbb{E}\left\{\cdot\right\}$ represnts the operation of taking expectation and the $n$-th main diagonal element is transmit power of the $n$-th subcarrier, and $\alpha=(T_Z+2\delta)/T_D$.

Singular value decomposition (SVD) based pre-coding and pre-decoding can be used to restore the desired symbols and achieve the capacity given by \eqref{C_S}~\cite{Fundamentals_of_Wireless_Communication}. The singular value decomposition (SVD) of $\mathbf{V}_{i,m}\mathbf{H}_{3-i,i}$ can be given by
\begin{equation}\label{SVD}
\begin{split}
\mathbf{ V}_{i,m}\mathbf{ H}_{3-i,i}=\mathbf{\tilde U}_{i,m}\mathbf{S}_{i,m}\mathbf{\hat U}_{i,m}^H,
\end{split}
\end{equation}
where $\mathbf{\tilde U}_{i,m}\in\mathbb{C}^{G \times G}$ and $\mathbf{\hat U}_{i,m}\in\mathbb{C}^{2N\times2N}$ are two unitary matrices corresponding to $\mathbf{ V}_{i,m}\mathbf{H}_{3-i,i}$, $(\cdot)^H$ denotes the conjugate transpose operation, and $\mathbf{S}_{i,m}\in\mathbb{C}^{G\times2N}$ is the matrix with main diagonal entries equaling to the singular values of $\mathbf{V}_{i,m}\mathbf{H}_{3-i,i}$ and other entries equaling to zero.
With the pre-coding of $\mathbf{X}_{3-i,m}=\mathbf{\hat U}_{i,m}\mathbf{\tilde X}_{3-i,m}$
and the pre-decoding of $\mathbf{\tilde Y}_{i,m}=\mathbf{\tilde U}_{i,m}^H\mathbf{Y}_{i,m}$,
\eqref{Y_ik_new} can be converted into multiple parallel SISO transmissions:
\begin{equation}\label{tilde_Y_i}
\begin{split}
\tilde Y_{i}^{\tilde k,m}\!=\!\sqrt{\lambda_{i}^{\tilde k,m}}\tilde X_{3-i}^{\tilde k,m}\!+\!\tilde Z_{i}^{\tilde k,m},~\tilde k=1,2,...,2N.
\end{split}
\end{equation}
In \eqref{tilde_Y_i}, $\tilde Y_{i}^{\tilde k,m}$ is the $\tilde k$-th entry of $\mathbf{\tilde Y}_{i,m}$, $\sqrt{\lambda_{i}^{\tilde k,m}}$ is the $\tilde k$-th singular value of $\mathbf{V}_{i,m}\mathbf{H}_{3-i,i}$, $\tilde X_{3-i}^{\tilde k,m}$ is the $\tilde k$-th entry of $\mathbf{\tilde X}_{3-i,m}$, whose average power can be given by $\mathbb{E}\{|\tilde X_{3-i}^{\tilde k,m}|^2\}=P_{3-i}^{\tilde k}$, and $\tilde Z_{i,k}$ is the $\tilde k$-th entry of $\mathbf{\tilde U}_{i,m}^H\mathbf{Z}_{i,m}$. Also, $P_{i}^{\tilde k}$ satisfies $\sum_{\tilde k=1}^{2N}P_{i}^{\tilde k}\leq \bar{P}_{i}$, where $\bar{P}_{i}$ denotes the power constraint for $U_i$.
Then, the estimation for $\tilde{X}_{3-i}^{\tilde k,m}$, denoted by $\hat{X}_{3-i}^{\tilde k,m}$, can be expressed as follows:
\begin{equation}\label{tilde_X_i_Estimate}
\begin{split}
\bar X_{3-i}^{\tilde k,m}=\underset{{X\in\mathbb{X}_{3-i}}}{\arg\min}\left|\frac{\tilde Y_{i}^{\tilde k,m}}{\sqrt{\lambda_{i}^{\tilde k,m}}}- X\right|^2,
\end{split}
\end{equation}
where $\mathbb{X}_{3-i}$ is the constellation that defines the transmit symbols of $U_{3-i}$. MIMO detection methods such as zero-forcing (ZF) detection, minimum mean square error (MMSE) detection, and MMSE with successive interference cancellation (SIC) can also be used to restore desired symbols~\cite{Fundamentals_of_Wireless_Communication,MIMO_OFDM}. We omit them for brevity here.
Based on \eqref{tilde_Y_i}, the signal to noise ratio (SNR) corresponding to the $k$-th parallel channel of $U_i$, denoted by $\gamma_{i}^{\tilde k,m}$, can be given by
\begin{equation}\label{SINR_n}
\begin{split}
\gamma^{\tilde k,m}_i&\!=\!\frac{P_{3-i}^{\tilde k}\lambda_{i}^{\tilde k,m}}{\sigma_0^2}.
\end{split}
\end{equation}
If quadrature phase shift keying (QPSK) is used, the bit error rate (BER) corresponding to the $\tilde k$-th parallel channel of $U_i$, denoted by $p_{b,i}^{\tilde k,m}$, can be given by~\cite{Wireless_Communications}
\begin{equation}\label{p_b}
\begin{split}
p_{b,i}^{\tilde k,m}=Q\left(\sqrt{\gamma^{\tilde k,m}_i}\right).
\end{split}
\end{equation}
\begin{figure*}[ht]
\begin{equation}\label{C_S_MIMO}
\begin{split}
C_{S}^{\rm MIMO}\!\!=\!\!\frac{\Delta f}{(1\!+\!\alpha)M}\sum_{m=1}^{M}\left[\log _2 \!\det\!\left(\!\mathbf{I}_{K_{R,1}G}\!\!+\!\!\frac{\mathbf{V}_{1,m}^{\rm MIMO}\mathbf{ R}_{2,m}^{\rm MIMO}(\mathbf{V}_{1,m}^{\rm MIMO})^H}{\sigma_0^2}\right)\!\!+\!\!\log _2 \!\det\!\left(\!\mathbf{I}_{K_{R,2}G}\!\!+\!\!\frac{\mathbf{V}_{2,m}^{\rm MIMO}\mathbf{ R}_{1,m}^{\rm MIMO}(\mathbf{V}_{2,m}^{\rm MIMO})^H}{\sigma_0^2}\right)\right]
\end{split}
\end{equation}
\hrulefill
\end{figure*}

The value of $\alpha$, which depends on the value of $T_Z$, has significant impacts on the singular values of the equivalent channel. If $\alpha$ is small and the resulting ${T}_{C,i}$ is short, we can obtain a high symbol rate, which is $\Delta f/(1+\alpha)$ in \eqref{C_S}. However, in this case, some values of $\lambda_{i}^{1,m},...,\lambda_{i}^{2N,m}$ are small, that is, some subchannels are weak, since the samples are close, which results in the low capacity and high BER for the corresponding weak subchannels. On the other hand, if $T_Z$ is large and the resulting ${T}_{C,i}$ is long, the values of $\lambda_{i}^{1,m},...,\lambda_{i}^{2N,m}$ are uniform, which results in fewer weak subchannels. However, the large value of $\alpha$ results in a relatively low symbol rate. Thus, $\alpha$ not only affects the symbol rate but also the equivalent channel. The optimal $\alpha$ can be numerically found with one-dimensional search methods such as golden section and Fibonacci search~\cite{Intro_Optimization}.

\section{MIMO ZIMS-VFD}
In this section, we discuss how to extend the SISO ZIMS-VFD to MIMO communications.

If $\tau_{i, j}^{\max}$ and $\tau_{i, j}^{\min}$ in Section III are replaced with $\tau_{i,k}^{\max}\!=\!\max_{l,p,q}\{\tau_{i,k,p,q}^l\}$ and $\tau_{i,k}^{\min}\!=\!\min_{l,p,q}\{\tau_{i,k,p,q}^l\}$, respectively, \eqref{T_F1}-\eqref{Condition_Total} also hold for MIMO ZIMS-VFD. Therefore, candidate intervals determination, receiver switch control, and sampling are similar to those in SISO ZIMS-VFD. Next we analyze the symbol restoration for MIMO ZIMS-VFD. We denote by $\mathbf{X}_{i,p}^{m}\in\mathbb{C}^{2N\times1}$ the transmit symbol vector corresponding to the $p$-th transmit antenna of $U_i$ in the $m$-th OFDM symbol period, $\mathbf{Y}_{i,q}^{m}\in\mathbb{C}^{G\times1}$ the sampling sequence corresponding to the $q$-th receive antenna of $U_i$ in the $m$-th candidate interval, $\mathbf{H}_{i,k}^{p,q}\in\mathbb{C}^{2N\times2N}$ the frequency-domain channel matrix from the $p$-th transmit antenna of $U_{i}$ to the $q$-th receive antenna of $U_k$. Then, the receive signal vector of $U_i$, denoted by $\mathbf{Y}_{i,m}^{\rm MIMO}\in\mathbb{C}^{K_{R,i}G\times1}$, can be given by
\begin{equation}\label{Y_ik_new_MA}
\begin{split}
\mathbf{Y}_{i,m}^{\rm MIMO}=\mathbf{V}_{i,m}^{\rm MIMO}\mathbf{ X}_{3-i,m}^{\rm MIMO}+\mathbf{ Z}_{i,m}^{\rm MIMO}.
\end{split}
\end{equation}
In \eqref{Y_ik_new_MA}, $\mathbf{Y}_{i,m}^{\rm MIMO}$ is given by $\mathbf{Y}_{i,m}^{\rm MIMO}=\left[\mathbf{Y}_{i,1}^{m},...,\mathbf{Y}_{i,K_{R,i}}^{m}\right]^T$, $\mathbf{ X}_{3-i,m}^{\rm MIMO}\in\mathbb{C}^{2K_{T,3-i}N\times1}$ is the transmit symbol vector of $U_{3-i}$ and given by $\mathbf{ X}_{3-i,m}^{\rm MIMO}=\big[\mathbf{X}_{3-i,1}^{m},...,\mathbf{X}_{3-i,K_{T,3-i}}^{m}\big]^T$, $\mathbf{V}_{i,m}^{\rm MIMO}\in\mathbb{C}^{K_{R,i}G\times2K_{T,3-i}N}$ is the equivalent channel matrix given by
\begin{equation}\label{V_i_MA}
\mathbf{V}_{i,m}^{\rm MIMO}=\left[\begin{array}{cccc}
\mathbf{V}_{i,m}\mathbf{H}_{3-i,i}^{1,1} & \cdots &\mathbf{V}_{i,m}\mathbf{H}_{3-i,i}^{K_{T,3-i},1} \\
\vdots  & \ddots & \vdots \\
\mathbf{V}_{i,m}\mathbf{H}_{3-i,i}^{1,K_{R,i}} & \cdots & \mathbf{V}_{i,m}\mathbf{H}_{3-i,i}^{K_{T,3-i},K_{R,i}}
\end{array}\right],
\end{equation}
where $\mathbf{V}_{i,m}\mathbf{H}_{3-i,i}^{p,q}$ is the equivalent channel matrix from the $p$-th transmit antenna of $U_{3-i}$ to the $q$-th receive antenna of $U_i$. We have the following theorem for $\mathbf{V}_{i,m}^{\rm MIMO}$.
\begin{theorem}
$\text{Rank}(\mathbf{V}_{i,m}^{\rm MIMO})=\text{Rank}(\mathbf{H}_{3-i,i}^{\rm MIMO})$,
where
\begin{equation}\label{H_i_MA}
\mathbf{H}_{3-i,i}^{\rm MIMO}=\left[\begin{array}{cccc}
\mathbf{H}_{3-i,i}^{1,1} & \cdots &\mathbf{H}_{3-i,i}^{K_{T,3-i},1} \\
\vdots  & \ddots & \vdots \\
\mathbf{H}_{3-i,i}^{1,K_{R,i}} & \cdots &\mathbf{H}_{3-i,i}^{K_{T,3-i},K_{R,i}}
\end{array}\right].
\end{equation}
\end{theorem}
\begin{proof}
Please refer to Appendix C.
\end{proof}

\textit{Remarks on Theorem 3:} Theorem 3 shows that $\mathbf{{V}}_{i,m}$ has no impact on the rank of $\mathbf{V}_{i,m}^{\rm MIMO}$, which implies that the sampling in candidate intervals has no impact on the rank of equivalent MIMO channel, and can achieve the same freedom degrees as conventional MIMO-OFDM communications.

Based on \eqref{Y_ik_new_MA}, the sum capacity of MIMO ZIMS-VFD can be given by \eqref{C_S_MIMO}. In \eqref{C_S_MIMO}, $\mathbf{I}_{K_{R,i}G}$ is a $K_{R,i}G\times K_{R,i}G$ identity matrix, $\mathbf{R}_{3-i,m}^{\rm MIMO}=\mathbb{E}\{\mathbf{X}_{3-i,m}^{\rm MIMO}(\mathbf{X}_{3-i,m}^{\rm MIMO})^H\}$. Then, similar to single-antenna case, based on the SVD of $\mathbf{V}_{i,m}^{\rm MIMO}$, \eqref{C_S_MIMO} can be achieved and the desired symbols can be restored. MIMO detection methods such as ZF detection, MMSE detection, and MMSE-SIC detection can also be used to restore the desired symbols~\cite{Fundamentals_of_Wireless_Communication,MIMO_OFDM}.
\section{Numerical Evaluations}
In this section, numerical evaluations are conducted to show the performance of ZIMS-VFD. The general parameters in evaluations are as follows. The maximum delay spread is set to $\tau_{\max}=100$~ns, the upper-bound of transitions is set to $\delta=1.9$~$\mu$s~\cite{RF_Switch,RF_Switch_1}, the central frequency is set to $f_c=2.4$~GHz, the bandwidth is set to $B=20$~MHz, the channel gains of SI channels follow $H_{i,i}^{p,q}\sim\mathcal C \mathcal N(0,1)$, the channel gains of desired channels follow $H_{3-i,i}^{p,q}\sim\mathcal C \mathcal N(0,10^{-10})$, the noise power is set to $\sigma_0^2=N_0B$ where $N_0$ is set to $-150$~dBm/Hz, the number of samples in each candidate interval is $2N$.
\begin{figure}
\centering
\includegraphics[width=0.5\textwidth]{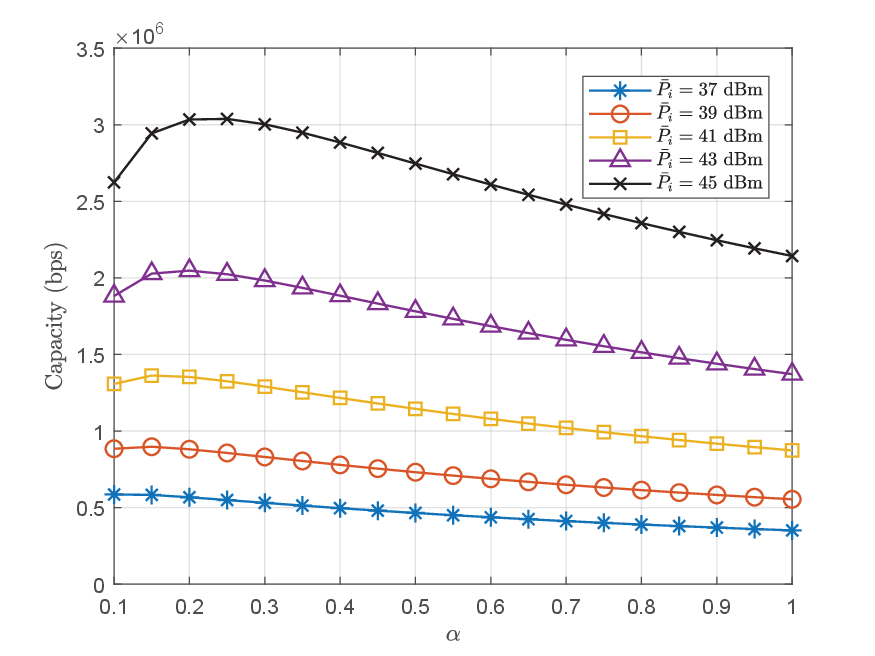}
\caption{The capacity of SISO ZIMS-VFD with different values of $\alpha$ and $\bar {P}_i$.}
 \label{fig:sum_rate}
\end{figure}

Figure~\ref{fig:sum_rate} shows the capacity of SISO ZIMS-VFD with different values of $\alpha$ and $\bar {P}_i$, where the number of subcarriers is set to $2N=2048$ and the power constraint for each transceiver is set to $\bar P_i=37$~dBm, $\bar P_i=39$~dBm, $\bar P_i=41$~dBm, $\bar P_i=43$~dBm, and $\bar P_i=45$~dBm. The transmit power of each transceiver is averagely allocated to subcarriers and the value of $\alpha$ ranges from 0.1 to 1. It can be known from Fig.~\ref{fig:sum_rate} that if the power constraint is relatively small, for example, $\bar P_i=37$~dBm, the capacity of SISO ZIMS-VFD monotonically decreases as $\alpha$ increases. If the power constraint is relatively large, for example, $\bar P_i=41$~dBm, the capacity of SISO ZIMS-VFD doesn't monotonically varies as $\alpha$ increases. Therefore, in low power region, the value of $\alpha$ needs to be chosen as small as possible. In high power region, the minimum $\alpha$ may not be optimal. The optimal $\alpha$ can be numerically found with one-dimensional search methods such as golden section search and Fibonacci search~\cite{Intro_Optimization}.
\begin{figure}\label{fig:sum_rate_gain}
\centering
\includegraphics[width=0.5\textwidth]{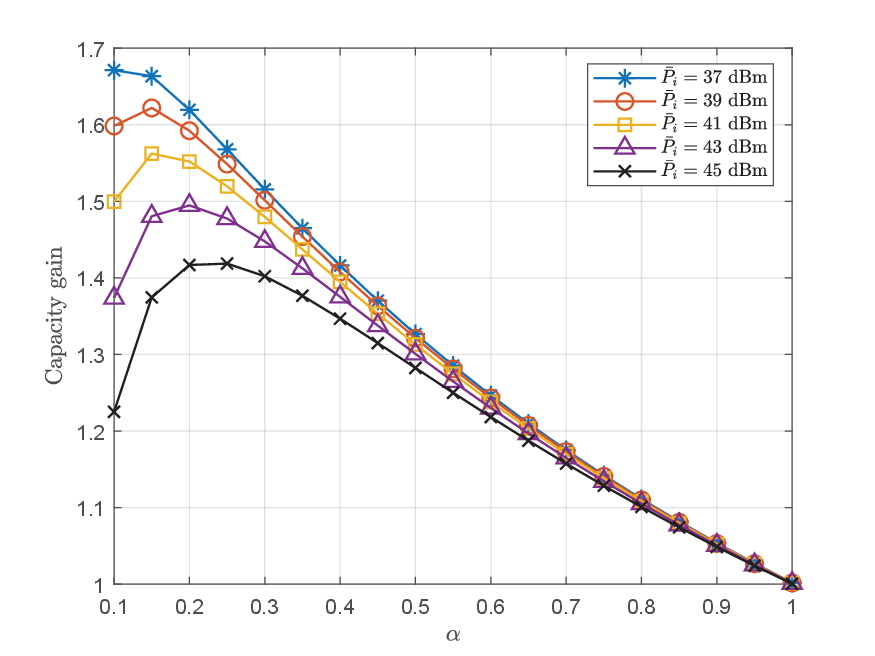}
\caption{The capacity gain of SISO ZIMS-VFD over SISO OFDM-HD with different values of $\alpha$ and $\bar {P}_i$.} \label{fig:sum_rate_gain}
\end{figure}

Figure~\ref{fig:sum_rate_gain} shows the capacity gain of SISO ZIMS-VFD over conventional SISO OFDM-HD, which refers to the ratio of the capacity of SISO ZIMS-VFD to the capacity of conventional SISO OFDM-HD, with different values of $\alpha$ and $\bar {P}_i$. The CP length of SISO OFDM-HD is set to $2\tau_{\max}$. The number of subcarriers, power constraint for each transceiver, power allocation and the range of $\alpha$ are the same as Fig.~\ref{fig:sum_rate}. We can observe from Fig.~\ref{fig:sum_rate_gain} that the capacity gain is larger than 1, which implies that SISO ZIMS-VFD outperforms conventional SISO OFDM-HD in capacity. Also, the capacity gain decreases as $\bar{P}_{i}$ increases, which implies that SISO ZIMS-VFD has the better performance gain in low-power region.
In addition, the optimal $\alpha$ to maximize the capacity gain monotonically increases as $\bar{P}_{i}$ increases. Moreover, we can see that the capacity gain is approximately equal to 1 if $\alpha=1$, that is, $T_Z=T_D-2\delta$. This is because the forward transmission and the backward transmission of ZIMS-VFD are completely separated due to the relatively long zero-interval.
\begin{figure}\label{fig:OFDM_Capacity_N}
\centering
\includegraphics[width=0.5\textwidth]{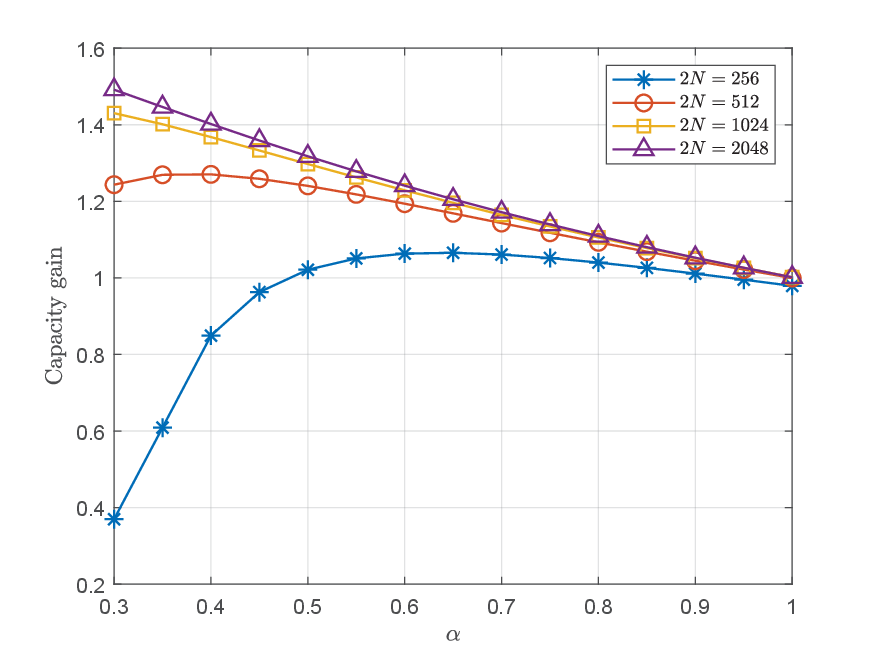}
\caption{The capacity gain of SISO ZIMS-VFD over conventional SISO OFDM-HD with different values of $\alpha$ and $2N$.} \label{fig:OFDM_Capacity_N}
\end{figure}

Figure~\ref{fig:OFDM_Capacity_N} shows the capacity gain of SISO ZIMS-VFD over SISO OFDM-HD with different values of $\alpha$ and $2N$. The CP length of SISO OFDM-HD is set to $2\tau_{\max}$. The number of subcarriers is set to $2N=256$, $2N=512$, $2N=1024$, and $2N=2048$. The power constraint for each transceiver is 40~dBm and the power is averagely allocated for subcarriers. We set the range of $\alpha$ is from 0.3 to 1. It can be known from Fig.~\ref{fig:OFDM_Capacity_N} that the capacity gain monotonically increases as the value of $2N$ increases, which implies that SISO ZIMS-VFD is suitable for OFDM transmission with a large number of subcarriers. Also, we can observe that if the subcarrier number is relatively large, for example, $2N=2048$, the capacity of SISO ZIMS-VFD monotonically decreases as $\alpha$ increases. If the subcarrier number is relatively small, for example, $2N=512$, the capacity of SISO ZIMS-VFD doesn't monotonically varies as $2N$ increases. Therefore, if the number of subcarriers is large, the value of $\alpha$ needs to be chosen as small as possible. Otherwise, the minimum $\alpha$ may not be optimal and the optimal $\alpha$ can be numerically found with one-dimensional search methods. Moreover, similar to Fig.~\ref{fig:sum_rate}, we can see that the capacity in Fig.~\ref{fig:OFDM_Capacity_N} is approximately equal to 1 due to the relatively long zero-interval.
\begin{figure}\label{fig:OFDM_Capacity_MIMO}
\centering
\includegraphics[width=0.5\textwidth]{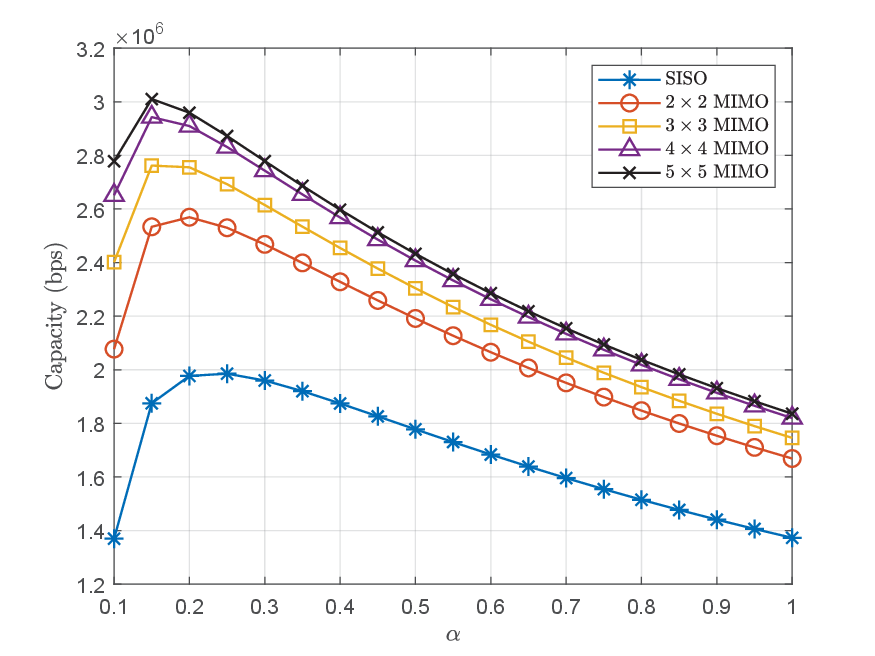}
\caption{The capacity ZIMS-VFD with different numbers of transmit antennas and receive antennas.} \label{fig:OFDM_Capacity_MIMO}
\end{figure}

Figure~\ref{fig:OFDM_Capacity_MIMO} compares the capacities of SISO ZIMS-VFD, $2\times2$ MIMO ZIMS-VFD, $3\times3$ MIMO ZIMS-VFD, $4\times4$ MIMO ZIMS-VFD, and $5\times5$ MIMO ZIMS-VFD. The number of subcarriers is set to $2N=1024$ and the power constraint for each transceiver is 40~dBm, which is averagely allocated for subcarriers. We set the range of $\alpha$ is from 0.1 to 1. It can be known from Fig.~\ref{fig:OFDM_Capacity_MIMO} that the capacity monotonically increases as the numbers of transmit antennas and receive antennas increase. This is because of the increase in parallel subchannels. However, we can observe that if the subcarrier number is relatively large, the increase in the number of antennas has limited impact on the capacity. In this case, increasing transmit power is more effective than increasing the number of antennas.
\begin{figure}
\centering
\includegraphics[width=0.5\textwidth]{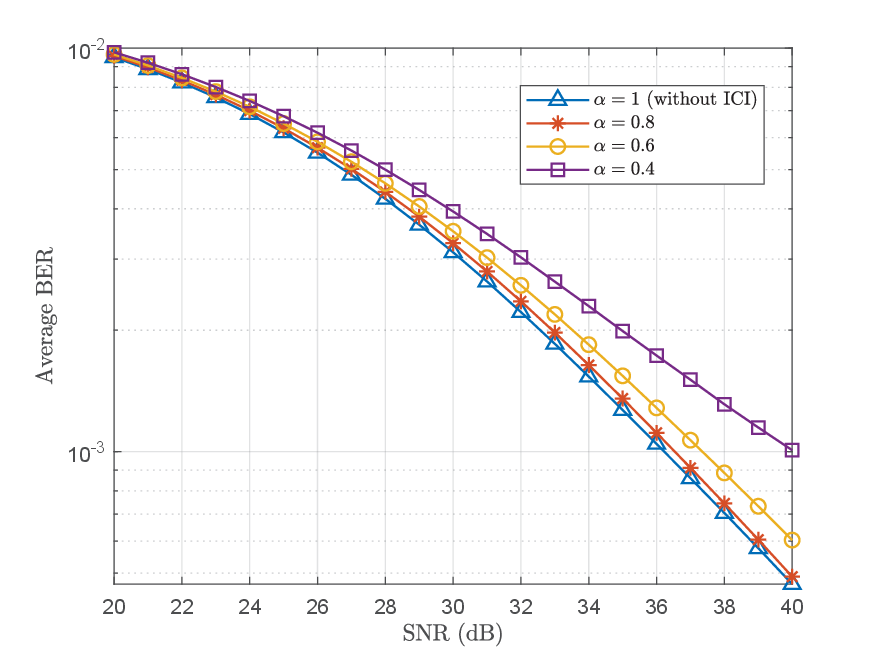}
\caption{Average BER of ZIMS-VFD.}
 \label{fig:Error_Power}
\end{figure}

Figure~\ref{fig:Error_Power} shows the average BER of ZIMS-VFD, where $2N=1024$, $K_{T,i}=K_{R,i}=2$, the modulation scheme is
QPSK, the values of $\alpha$ are 1, 0.8, 0.6, 0.4, and the SNR, whose expression is $\bar{P}_i/(N_0B)$, ranges from $20$~dB to $40$~dB. Also, MMSE detection is used for ZIMS-VFD to restore the desired symbols. If $\alpha=1$, there are no inter-carrier interferences (ICI) since the samples are uniformly distributed in the entire OFDM period and the desired symbols are restored with DFT. We can observe from Fig.~\ref{fig:Error_Power} that the average BER increases as the value of $\alpha$ decreases, which implies that the error performance decreases as the symbol rate increases. It can also be observed that if the SNR is relatively low, the average BERs corresponding to $\alpha=0.8$, $\alpha=0.6$, and $\alpha=0.4$ are very close to that corresponding to $\alpha=1$. If the SNR is relatively high, the average BERs corresponding to $\alpha=0.8$, $\alpha=0.6$, and $\alpha=0.4$ do not increase much as compared to that of $\alpha=1$. Thus, the error performance of ZIMS-VFD is not much different from the case without ICI and is acceptable in practical communication systems.
\begin{figure}\label{fig:SINR_Gain_Xi_P}
\centering
\includegraphics[width=0.5\textwidth]{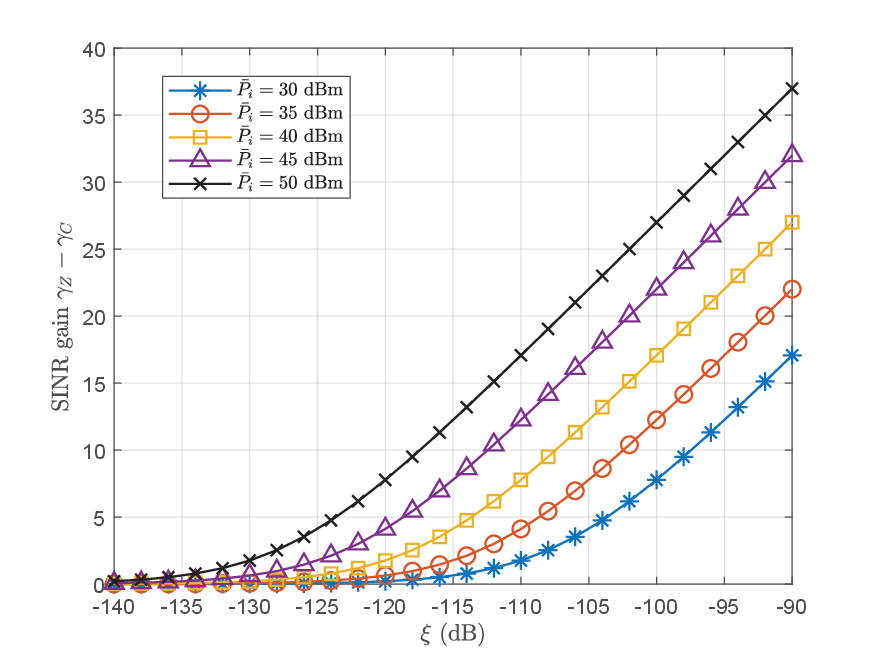}
\caption{The SINR gain of ZIMS-VFD over conventional FD with different values of SIC capability and $\bar {P}_i$.} \label{fig:SINR_Gain_Xi_P}
\end{figure}

Figure~\ref{fig:SINR_Gain_Xi_P} shows the SINR gain of ZIMS-VFD over conventional FD with SIC, which is equal to $\gamma_Z-\gamma_C$, where $\gamma_Z$ is the SNR (in dB) of ZIMS-VFD and $\gamma_C$ is the SINR (in dB) of conventional FD with SIC. In this paper, conventional FD with SIC refers to the FD system where the residual SI follows Gaussian distribution and its power is proportional to its transmit power. This model has been considered as a reasonable approximation in a number of existing works~\cite{Xiaofei_Xu_TVT_2017_17,Xiaofei_Xu_TVT_2017}. In this evaluation, we consider the BER performance of $U_1$. The expression of $\gamma_Z$ is
\begin{equation}\label{gamma_z}
\begin{split}
\gamma_Z=\frac{\sum_{\tilde k=1}^{4N}P_{2}^{\tilde k} \lambda_{1}^{\tilde k,m}}{4N\sigma_0^2}.
\end{split}
\end{equation}
Based on the residual SI model, the expression of $\gamma_C$ can be given by

\begin{equation}\label{gamma_c}
\begin{split}
\gamma_C=\frac{\sum_{\tilde k=1}^{4N}P_{2}^{\tilde k} \tilde H_{2,1}^{\tilde k}}{4N\sigma_0^2+\sum_{\tilde k=1}^{4N}\xi P_{1}^{\tilde k}}.
\end{split}
\end{equation}
where $\tilde H_{2,1}^{\tilde k}$ is the channel gain of the $\tilde k$-th subchannel and $\xi$ ($0\leq\xi\leq1$) is a parameter to characterize the SIC capability. If $\xi = 0$, the SI can be perfectly canceled. If $\xi = 1$, the SIC technique is invalid.
The number of subcarriers is set to $2N=1024$. The number of transmit antennas and receive antennas are set to $K_{T,i}=K_{R,i}=2$. The power constraint for each transceiver is set to $\bar P_i=30$~dBm, $\bar P_i=35$~dBm, $\bar P_i=40$~dBm, $\bar P_i=45$~dBm, and $\bar P_i=50$~dBm, which are averagely allocated to subchannels. It can be observed that the SINR gain increases as the value of $\xi$ increases, which implies that as compared with conventional FD with SIC, ZIMS-VFD has the better SINR performance if the SIC capability of conventional FD is relatively low. For example, if $\xi=-90$~dB and $\bar P_i=50$~dBm, the SINR gain of ZIMS-VFD over conventional FD with SIC is more than 35~dB. If the SIC capability of conventional FD is relatively high, for example, $\xi=-140$~dB, it has the similar performance with ZIMS-VFD. Also, we can see that the SINR gain increases as the value of $\bar P_i$ increases. This is because the residual SI power of conventional FD increases as the value of $\bar P_i$ increases, but ZIMS-VFD can avoid SI by sampling the receive signal in the SI-free intervals.
\begin{figure}
\centering
\includegraphics[width=0.5\textwidth]{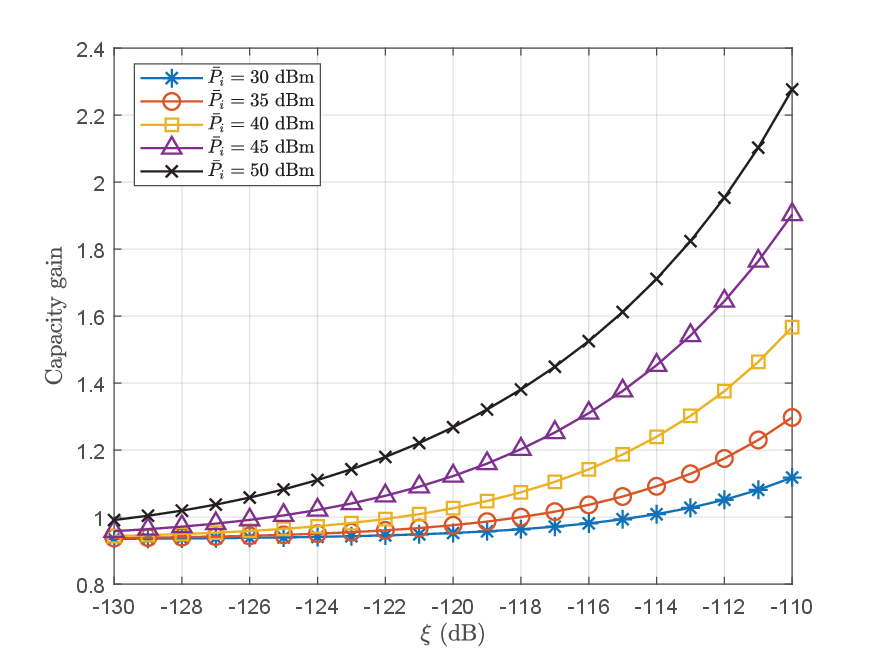}
\caption{The capacity gain of ZIMS-VFD over conventional FD with SIC under different values of $\xi$ and $\bar {P}_i$.}
 \label{fig:Sum_rate_Gain_Xi_P}
\end{figure}

Figure~\ref{fig:Sum_rate_Gain_Xi_P} shows the capacity gain of ZIMS-VFD over conventional FD with SIC, which refers to the ratio of the capacity of ZIMS-VFD to the capacity of conventional FD with SIC, under different values of $\xi$ and $\bar {P}_i$. The parameter settings are the same as those in Fig.~\ref{fig:SINR_Gain_Xi_P}.
It can be observed that the capacity gain increases as the value of $\xi$ increases, which implies that if the SIC capability of conventional FD is relatively weak, the capacity gain of ZIMS-VFD over conventional FD with SIC is relative large. For example, if $\xi=-110$~dB and $\bar P_i=50$~dBm, the capacity gain of ZIMS-VFD over conventional FD with SIC is more than 2. If the SIC capability of conventional FD is relatively strong, for example, $\xi=-130$~dB, it has the better performance than ZIMS-VFD. Also, since ZIMS-VFD can avoid SI instead of directly canceling SI in a certain proportion, the capacity gain increases as the value of $\bar P_i$ increases.
\begin{figure}
\centering
\includegraphics[width=0.5\textwidth]{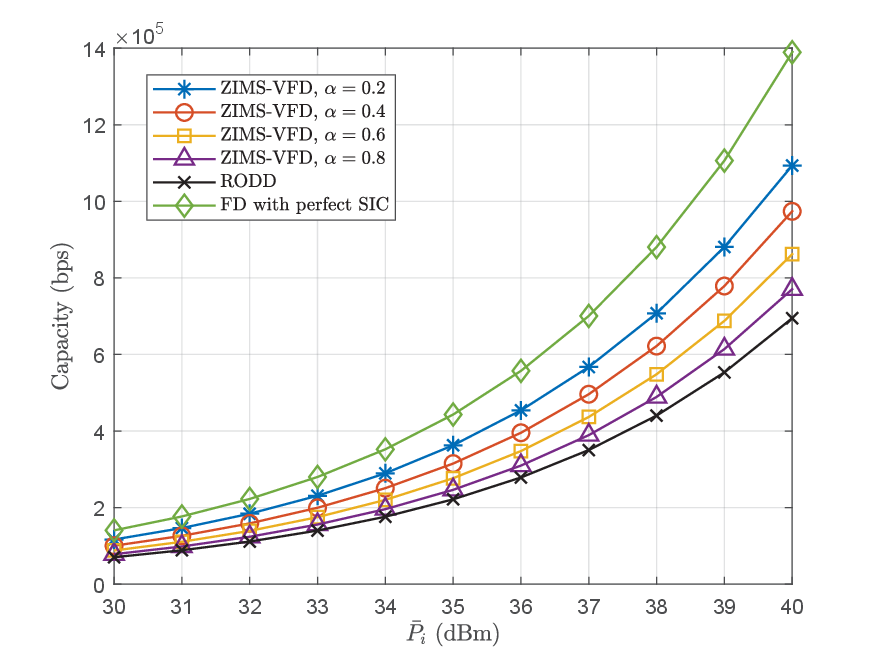}
\caption{The capacities of ZIMS-VFD, FD with perfect SIC, and RODD.}
 \label{fig:sum_rate_CMP_optimal}
\end{figure}

Figure~\ref{fig:sum_rate_CMP_optimal} compares the capacities of ZIMS-VFD, FD with perfect SIC, and RODD, where the subcarrier number is set to $2N=1024$, the number of transmit antennas and receive antennas are set to $K_{T,i}=K_{R,i}=2$. To maximize the throughput of RODD, the duty cycle of RODD is set to 1/2 and the coding rate of Vandermonde Reed Solomon codes is set to 1/2~\cite{RODD3}. It can be known from Fig.~\ref{fig:sum_rate_CMP_optimal} that the capacity of ZIMS-VFD is always lower than that of FD with perfect SIC and always higher than that of RODD. If the length of zero-interval and the transmit power are relatively small, for example, $\alpha=0.2$ and $\bar{P}_i=30$ dBm, ZIMS-VFD is close to FD with perfect SIC. Also, with a given $\alpha$, we can see that the difference between ZIMS-VFD and FD with perfect SIC increases as the transmit power increases.
\section{Conclusion}
In this paper, we proposed a novel FD communication technique, ZIMS-VFD, where two HD transceivers can simultaneously transmit signals to each other and each transceiver can effectively restore the desired symbols. The SIC techniques in the propagation domain and analog domain can be reduced. Based on OFDM, the transmission mechanism of ZIMS-VFD was introduced. We designed the transmit signal structure and determined the candidate interval, where both the SI and transitions are zero and the desired signal is within the data-interval. We have derived the condition that the candidate interval is always larger than zero. Also, we showed that the receive signal can be sampled in candidate intervals to restore the desired symbols. Numerical results verified our theoretical analyses and showed that ZIMS-VFD can effectively avoid SI and increase the capacity.
\section*{Appendix A \\Proof for Theorem 1}
It can be calculated from \eqref{XU1} and \eqref{XL1} that the possible vales of ${T}_{C,1}^m$ are
$T_Z-2\delta+\tau_{1,1}^{\min}-\tau_{1,1}^{\max}$,
$T_D\!+\!\tau_{1,1}^{\min}\!-\!\tau_{2,1}^{\min}$,
${T}_{Z}-2\delta+\tau_{2,1}^{\max}\!-\!\tau_{1,1}^{\max}$,
and
${T}_{D}+\tau_{2,1}^{\max}\!-\!\tau_{2,1}^{\min}$.
Due to \eqref{T_D_Condition}, we have $T_D\!+\!\tau_{1,1}^{\min}\!-\!\tau_{2,1}^{\min}>0$ and ${T}_{D}+\tau_{2,1}^{\max}\!-\!\tau_{2,1}^{\min}>0$. To guarantee ${T}_{C,1}^m>0$, $T_Z$ needs to satisfy
\begin{equation}\label{apdx1}
\begin{split}
T_Z>2\delta-\tau_{1,1}^{\min}+\tau_{1,1}^{\max}
\end{split}
\end{equation}
and
\begin{equation}\label{apdx2}
\begin{split}
{T}_{Z}>2\delta-\tau_{2,1}^{\max}\!+\!\tau_{1,1}^{\max}.
\end{split}
\end{equation}
On the other hand, based on \eqref{XU2} and \eqref{XL2}, the possible vales of ${T}_{C,2}^m$ are
$T_Z-2\delta+\tau_{2,2}^{\min}-\tau_{2,2}^{\max}$,
$T_Z\!-\!2\delta+\tau_{2,2}^{\min}\!-\!\tau_{1,2}^{\min}$,
$T_D+\tau_{1,2}^{\max}\!-\!\tau_{2,2}^{\max}$,
and
${T}_{D}+\tau_{1,2}^{\max}\!-\!\tau_{1,2}^{\min}$.
Due to \eqref{T_D_Condition}, we have $T_D+\tau_{1,2}^{\max}\!-\!\tau_{2,2}^{\max}>0$ and ${T}_{D}+\tau_{1,2}^{\max}\!-\!\tau_{1,2}^{\min}>0$.
To guarantee ${T}_{C,2}^m>0$, $T_Z$ needs to satisfy
\begin{equation}\label{apdx3}
\begin{split}
T_Z>2\delta-\tau_{2,2}^{\min}+\tau_{2,2}^{\max}
\end{split}
\end{equation}
and
\begin{equation}\label{apdx4}
\begin{split}
T_Z>\!\!2\delta-\tau_{2,2}^{\min}\!+\!\tau_{1,2}^{\min}.
\end{split}
\end{equation}
Since the difference between any two delays is less than $\tau_{\max}$, if \eqref{thm} is satisfied, \eqref{apdx1}-\eqref{apdx4} hold and all possible values of ${T}_{C,i}^m$ are always larger than 0. Thus, Theorem 1 follows.
\section*{Appendix B\\Proof for Theorem 2}
We can write $\mathbf{ V}_{i,m}$ as \eqref{V_i_Apdx}, which is at the top of the next page.
We can observe from \eqref{V_i_Apdx} that $\mathbf{V}_{i,m}$ is the product of a diagonal matrix, whose rank is $G$, and a Vandermond matrix, whose rank is $2N$. Thus, we have $\text{Rank}(\mathbf{V}_{i,m})\le \min\{G,2N\}=2N$. Also, according to Sylvester's inequality, we have $\text{Rank}(\mathbf{V}_{i,m})\ge G+2N-G=2N$. Therefore, $\text{Rank}(\mathbf{V}_{i,m})=2N$. Furthermore, since $\text{Rank}(\mathbf{ H}_{3-i,i})=2N$, following the similar steps we can conclude that $\text{Rank}(\mathbf{ V}_{i,m}\mathbf{ H}_{3-i,i})=2N$. Theorem 2 follows.
\section*{Appendix C\\Proof for Theorem 3}
We can write $\mathbf{V}_{i,m}^{\rm MIMO}$ as
\begin{equation}\label{V_i_MA_apdx}
\mathbf{V}_{i,m}^{\rm MIMO}=\left[\begin{array}{cccc}
\mathbf{V}_{i,m} &  & \\
  & \ddots &  \\
&  & \mathbf{V}_{i,m}
\end{array}\right]\mathbf{H}_{3-i,i}^{\rm MIMO}.
\end{equation}
\begin{figure*}
\begin{equation}\label{V_i_Apdx}
\begin{split}
\mathbf{V}_{i,m}&=\left[\begin{array}{cccc}
e^{j2\pi f_{-N+1}\hat{t}_{i}^{1,m}} & &   \\
   & \ddots&  \\
   &    & e^{j2\pi f_{-N+1}\hat{t}_{i}^{G,m}}
\end{array}\right]
\times\left[\begin{array}{cccc}
1 &e^{j2\pi \Delta f\hat{t}_{i}^{1,m}}&  \cdots & (e^{j2\pi \Delta f\hat{t}_{i}^{1,m}})^{2N} \\
1 &e^{j2\pi \Delta f\hat{t}_{i}^{2,m}}&  \cdots & (e^{j2\pi \Delta f\hat{t}_{i}^{2,m}})^{2N} \\
\vdots  &\vdots& \ddots & \vdots \\
1 &e^{j2\pi \Delta f\hat{t}_{i}^{G,m}}&  \cdots & (e^{j2\pi \Delta f\hat{t}_{i}^{G,m}})^{2N}
\end{array}\right].
\end{split}
\end{equation}
\hrulefill
\vspace{-15pt}
\end{figure*}It can be observed that $\mathbf{V}_{i,m}^{\rm MIMO}$ is the product of a block diagonal matrix, whose rank is $2NK_{R,i}$, and $\mathbf{H}_{3-i,i}^{\rm MIMO}$. Thus, we have $\text{Rank}(\mathbf{V}_{i,m}^{\rm MIMO})\leq\text{Rank}(\mathbf{H}_{3-i,i}^{\rm MIMO})$. In addition, according to Sylvester's inequality, we have $\text{Rank}(\mathbf{V}_{i,m}^{\rm MIMO})\geq 2NK_{R,i}+\text{Rank}(\mathbf{H}_{3-i,i}^{\rm MIMO})-2NK_{R,i}=\text{Rank}(\mathbf{H}_{3-i,i}^{\rm MIMO})$. Based on above analyses, we can conclude that $\text{Rank}(\mathbf{V}_{i,m}^{\rm MIMO})=\text{Rank}(\mathbf{H}_{3-i,i}^{\rm MIMO})$. Theorem 3 follows.
\bibliographystyle{IEEEbib}
\bibliography{ref}
\end{document}